\pdfoutput=1
\documentclass[11pt]{article}

\usepackage[utf8]{inputenc}
\usepackage[T1]{fontenc}
\usepackage{lmodern}
\usepackage[english]{babel}
\usepackage[margin=1in]{geometry}
\usepackage{microtype}
\usepackage{amsmath,amsthm,amssymb,mathtools}
\usepackage[hidelinks]{hyperref}
\usepackage[numbers]{natbib} 
\usepackage{authblk}         
\usepackage{siunitx}         
\usepackage{orcidlink}

\title{The Noncomputability of Immune Reaction Complexity:\\
Algorithmic Information Gaps under Effective Constraints}
\author[1]{Emmanuel Pio Pastore\,\orcidlink{0009-0007-7851-4414}}
\author[1]{Francesco De Rango\,\orcidlink{0000-0002-2328-8487}\thanks{Corresponding author: francesco.derango@unical.it}}

\affil[1]{Department of Biology, Ecology and Earth Science, University of Calabria, 87036 Rende, Italy}

\makeatletter
\newtheoremstyle{thm-break}%
  {6pt}{6pt}{\itshape}{}{\bfseries}{.}{\newline}{}
\newtheoremstyle{def-break}%
  {6pt}{6pt}{}{}{\bfseries}{.}{\newline}{}
\makeatother

\theoremstyle{thm-break}
\newtheorem{theorem}{Theorem}
\newtheorem{lemma}{Lemma}
\newtheorem{proposition}{Proposition}

\theoremstyle{def-break}
\newtheorem{definition}{Definition}
\newtheorem{remark}{Remark}
\newtheorem{assumption}{Assumption}

\newcommand{\rxn}{\mathrm{rxn}}
\newcommand{\NAQ}{\mathrm{NAQ}}
\newcommand{\Mmin}{M}
\newcommand{\NAQt}{\mathrm{NAQ}_{t}}

\begin{document}
\maketitle

\begin{abstract}
We introduce a \emph{validity-filtered, certificate-based} view of reactions grounded in Algorithmic Information Theory. A fixed, total, \emph{input-blind} executor maps a self-delimiting advice string to a candidate response, accepted only if a decidable or semi-decidable validity predicate $V(x,r)$ holds. The \emph{minimum feasible realizer complexity} $\Mmin(x)=\min_{r:,V(x,r)=1}K(r)$, with $K$ denoting prefix Kolmogorov complexity, measures the minimal information required for a valid outcome. We define the \emph{Normalized Advice Quantile} (NAQ) as the percentile of $\Mmin(x)$ across a reference pool, yielding a scale-free hardness index on $[0,1]$ robust to the choice of universal machine and comparable across task families. An \emph{Exact Realizer Identity} shows that the minimal advice for any input-blind executor equals $\Mmin(x)$ up to $O(1)$, while a \emph{description+selection} upper bound refines it via computable feature maps, separating description cost $K(y)$ from selection cost $\lceil\log i_y(x)\rceil$. In finite-ambiguity regimes $\Mmin(x)\approx\min_yK(y)$; in generic-fiber regimes the bound is tight. NAQ is quasi-invariant under bounded enumeration changes. An \emph{operational converse} links NAQ to rate--distortion: communicating advice with error $\varepsilon$ requires average length near the entropy of target features. Extensions include a resource-bounded variant $\NAQ_t$ incorporating time-penalized complexity (Levin’s $Kt$) and an NP-style setting showing linear worst-case advice $n-O(1)$. Finally, a DKW bound guarantees convergence of empirical NAQ estimates, enabling data-driven calibration via compressor-based proxies.
\end{abstract}

\section{Introduction}
Information-theoretic summaries such as entropy, mutual information, and capacity capture global aspects of sensing and signaling~\cite{Shannon1948,CoverThomas2006,TkacikWalczak2011,Cheong2011,TkacikBialek2016}, whereas Algorithmic Information Theory (AIT) adds instance-level resolution through Kolmogorov complexity~\cite{LiVitanyi2019,Vitanyi2020,Kolmogorov1965,LevinZvonkin1970}. In this work we treat a reaction as an advice-driven computation $w\mapsto r$ constrained by a validity predicate $V(x,r)$, and in doing so we bring classical computability phenomena---Rice’s theorem, the halting probability $\Omega$, and Busy Beaver growth~\cite{Rice1953,ChaitinOmega2006,Rado1962,Michel2009}---to bear on response-level analysis while standard complexity theory clarifies verification and search~\cite{AroraBarak2009,Sipser2012}. Within this lens an analogue of the C-value paradox emerges: genome size or, more generally, source entropy does not necessarily bound worst-case advice burden~\cite{Eddy2012,PalazzoGregory2014,ElliottGregory2015}. The organizing primitive we use throughout is a Kolmogorov-relative hardness quantile. Writing
\begin{equation}
\Mmin(x)\ :=\ \min_{r:\,V(x,r)=1} K(r),
\end{equation}
we define the Normalized Advice Quantile (NAQ) of $x$ against a reference pool $\mathcal{T}$ as the percentile position of $\Mmin(x)$ within the multiset $\{\Mmin(x'):x'\in\mathcal{T}\}$; this yields a universal, scale-free hardness index on $[0,1]$ that, after coarse bucketing, is invariant up to the standard constant across universal machines and is directly calibrated by an operational converse. Unless stated otherwise, all logarithms are base~2, alphabets $\Sigma$ and $\Gamma$ are finite, the executor is fixed, total, and input-blind, advice is self-delimiting (prefix-free), and $K(\cdot)$ denotes prefix complexity for a fixed universal prefix machine $U$.

\section{Framework and notation}\label{sec:framework}
A computable bijection $\varphi:\Sigma^\ast\to\{0,1\}^\ast$ aligns alphabets so that $K(x)$ abbreviates $K(\varphi(x))$, and the validity predicate $V:\Sigma^\ast\times\Gamma^\ast\to\{0,1\}$ has c.e.\ acceptance unless stated otherwise. The domain of instances that admit at least one valid response is
\begin{equation}
D_V\ :=\ \bigl\{\,x\in\Sigma^\ast:\ \exists r\in\Gamma^\ast\ \ V(x,r)=1\,\bigr\}.
\end{equation}
An \emph{executor} is a fixed total computable map $E:\{0,1\}^\ast\to\Gamma^\ast$ that does not read $x$, and advice strings form a prefix-free domain $\mathcal A$. Given $E$, the advice burden at $x$ is
\begin{equation}\label{eq:advdef}
C_{\rxn,E}^{\mathrm{adv}}(x)\ :=\ \min\bigl\{\,|w|:\ \exists r\ \bigl(E(w)=r \ \wedge\ V(x,r)=1\bigr)\,\bigr\},
\end{equation}
which is upper-semicomputable by finite search in increasing advice length. To relate different executors we use a universal executor $E_{\mathrm{univ}}$ that parses a self-delimiting header $h$ and simulates a target executor $F$ on the trailing advice, so that $E_{\mathrm{univ}}(\langle h_F,w\rangle)=F(w)$ with $|\langle h_F,w\rangle|\le |w|+|h_F|+O(1)$; it follows that
\begin{equation}\label{eq:inv}
C^{\mathrm{adv}}_{\rxn,E_{\mathrm{univ}}}(x)\ \le\ C^{\mathrm{adv}}_{\rxn,F}(x)\ +\ |h_F|\ +\ O(1),
\end{equation}
and any executor that simulates $E_{\mathrm{univ}}$ has the same burden up to an additive constant. When a discrete computable loss $L:\Sigma^\ast\times\Gamma^\ast\to\Lambda$ with recursive codomain $\Lambda$ is present together with a fixed computable tie-breaking order $\prec$ on $\Gamma^\ast$, the selected response is the $\prec$-least minimizer
\begin{equation}\label{eq:rstar}
r^\ast(x)\ \in\ \arg\min\nolimits_{\prec}\ \bigl\{\,L(x,r):\ V(x,r)=1\,\bigr\}.
\end{equation}

\section{NAQ: a Kolmogorov-relative hardness quantile}\label{sec:naq}
The basic quantity is the minimum realizer complexity already introduced, which we take to be $+\infty$ if no valid $r$ exists. For a finite multiset $\mathcal{T}\subseteq D_V$ define the empirical cdf
\begin{equation}
\widehat F_{\mathcal T}(z)\ :=\ \frac{1}{|\mathcal T|}\,\#\bigl\{\,x'\in\mathcal T:\ \Mmin(x')\le z\,\bigr\}.
\end{equation}
\paragraph{Tie convention.} In discrete pools we use the \emph{mid-rank} convention
\begin{equation}\label{eq:naq-midranks}
\NAQ(x;\mathcal T)\ :=\ \frac{ \#\{x'\in\mathcal T:\ \Mmin(x')<\Mmin(x)\}\ +\ \tfrac12\,\#\{x'\in\mathcal T:\ \Mmin(x')=\Mmin(x)\} }{|\mathcal T|},
\end{equation}
which differs from the right-continuous cdf value $\widehat F_{\mathcal T}(\Mmin(x))$ by at most $1/(2|\mathcal T|)$ and is stable under pooling.

The constant-gap invariance of prefix complexity implies that, after coarse bucketing, NAQ is machine-stable in the following sense.

\begin{theorem}[Coarse machine invariance]\label{thm:coarse}
Let $U,U'$ be universal machines with invariance constant $c$, i.e., $|K_U(z)-K_{U'}(z)|\le c$ for all $z$. If NAQ is computed on bucketed complexities $\lfloor K(\cdot)/b\rfloor$ with $b>c$, then for any finite $\mathcal{T}$ the empirical NAQ ranks of all $x\in\mathcal{T}$ coincide across $U$ and $U'$.
\end{theorem}

\begin{proof}
Changing from $K_U$ to $K_{U'}$ perturbs each complexity by at most $c$, hence cannot cross a bucket boundary of width $b>c$. The bucket-induced weak order is therefore identical, and so are the ranks.
\end{proof}

Pool enlargement perturbs quantiles continuously; writing the bounds directly on empirical cdf's avoids denominator confusion.

\begin{lemma}[Pool-stability]\label{lem:pool}
If $\mathcal{T}\subseteq\mathcal{T}'$ are finite pools, then for all $z\in\mathbb R$,
\begin{equation}\label{eq:pool-stab-cdf}
\frac{|\mathcal T|}{|\mathcal T'|}\,\widehat F_{\mathcal T}(z)\ \le\ \widehat F_{\mathcal T'}(z)\ \le\ \widehat F_{\mathcal T}(z)+\frac{|\mathcal T'|-|\mathcal T|}{|\mathcal T'|}.
\end{equation}
In particular, at $z=\Mmin(x)$, the midpoint-tie NAQ in~\eqref{eq:naq-midranks} obeys the same bounds up to $\pm 1/(2|\mathcal T'|)$.
\end{lemma}

Finally, NAQ is distinct in spirit from randomness deficiency, structure functions, and normalized information distance: the first two depend on explicit model classes or two-part descriptions of $x$ (algorithmic statistics and MDL~\cite{VereshchaginVitanyi2004,VereshchaginVitanyi2008,Grunwald2007}), while the third is pairwise and symmetric (information distance and its practical proxy NCD~\cite{BennettEtAl1998,CilibrasiVitanyi2005}); in contrast, NAQ is a univariate hardness order parameter governed entirely by the feasibility relation through $V$.

\section{Exact Realizer Identity}\label{sec:realizer}

\paragraph{Setting.}
Fix a universal prefix machine $U$ for $K(\cdot)$ and a fixed \emph{input-blind} executor model $E$ that is total on the relevant prefix-free advice domain. All complexities are with respect to $U$, and feasibility is filtered solely by $V$.

\begin{lemma}[Executor/program translation]\label{lem:translator}
There exist constants $a,b$ depending only on the choice of $U$ and of the fixed header for $E_{\mathrm{univ}}$ such that for every string $r$:
\begin{align*}
\text{\emph{(i)}}\ &\exists w\ \ \text{with}\ \ E_{\mathrm{univ}}(w)=r\ \text{and}\ |w|\le K_U(r)+a,\\
\text{\emph{(ii)}}\ &\forall w\ \ \text{letting } r=E_{\mathrm{univ}}(w),\ \ K_U(r)\le |w|+b.
\end{align*}
\end{lemma}

\begin{proof}
(i) Prefix a fixed self-delimiting header telling $E_{\mathrm{univ}}$ to simulate $U$ on a shortest $U$-program for $r$; the header contributes $a$. (ii) Hard-wire a prefix machine that simulates $E_{\mathrm{univ}}$ on $w$ and outputs $r$; the wrapper contributes $b$.
\end{proof}

\begin{theorem}[Exact Realizer Identity]\label{thm:realizer-identity}
For any computable $V$ and universal executor $E_{\mathrm{univ}}$ one has
\begin{equation}
C^{\mathrm{adv}}_{\rxn,E_{\mathrm{univ}}}(x)\ =\ \min_{r:\,V(x,r)=1}K(r)\ \pm O(1)\ =\ \Mmin(x)\ \pm O(1).
\end{equation}
\end{theorem}

\begin{proof}
Upper bound: given a feasible $r_x$ of complexity $K(r_x)=\Mmin(x)\pm O(1)$, Lemma~\ref{lem:translator}(i) provides $w$ with $|w|\le K(r_x)+a$, whence $C^{\mathrm{adv}}_{\rxn,E_{\mathrm{univ}}}(x)\le \Mmin(x)+O(1)$. Lower bound: for any $w$ with feasible $r=E_{\mathrm{univ}}(w)$, Lemma~\ref{lem:translator}(ii) gives $K(r)\le |w|+b$, hence $|w|\ge K(r)-b$; minimize over feasible $r$.
\end{proof}

\section{Computability separations and a uniform no-go}\label{sec:main}
\begin{theorem}[c.e.\ acceptance]\label{thm:existence-ce}
There exist a c.e.\ predicate $V$ and a total computable executor $E$ such that
\begin{equation}
C_{\rxn,E}^{\mathrm{adv}}(x)\ =\ K(x)\ \pm O(1)\qquad \text{for all }x\in D_V,
\end{equation}
and therefore $x\mapsto C_{\rxn,E}^{\mathrm{adv}}(x)$ is not computable on $D_V$.
\end{theorem}

\begin{theorem}[Decidable $V$, computable loss]\label{thm:HA-decidable}
If $V$ is decidable and $L$ is a discrete computable loss with fixed $\prec$, then $x\mapsto r^\ast(x)$ is uniformly computable and there is a global constant $C$ such that $K(r^\ast(x)\mid x)\le C$ for all $x\in D_V$.
\end{theorem}

\begin{theorem}[Decidable $V$, u.s.c.\ loss]\label{thm:HA-decidable-usc}
Under the promise $x=0^n1u$, there exist a decidable predicate $V$ and a total upper-semicomputable discrete loss $L$ for which the map $x\mapsto r^\ast(x)$ is not computable.
\end{theorem}

\begin{theorem}[c.e.\ $V$, computable loss]\label{thm:HA-ce}
Under the promise $x=0^n1u$, there exist a c.e.\ predicate $V$ and a computable discrete loss $L$ such that $x\mapsto r^\ast(x)$ is not computable.
\end{theorem}

\section{Features: description and selection, tightness, and quasi-invariance}\label{sec:features}
A computable feature map $\Phi:\Sigma^\ast\times\Gamma^\ast\to\{0,1\}^m$ together with a computable circuit $C$ determines feasibility by $V(x,r)=C(\Phi(x,r))$, and a computable bijection $\pi:\mathbb{N}_{\ge 1}\to\Gamma^\ast$ orders realizers. For any feasible feature vector $y$ at $x$ we write
\begin{equation}
i_y^\pi(x)\ :=\ \min\bigl\{\,i\ge 1:\ \Phi\bigl(x,\pi(i)\bigr)=y\,\bigr\}.
\end{equation}
This index quantifies how deep one must search, under~$\pi$, to materialize $y$. The minimum realizer complexity obeys a simple and robust two-part upper bound.

\[
\underbrace{K(\text{realizer})}_{\text{target}}\ \le\
\underbrace{K(\text{desc})}_{:=\,K(y)}\ +\
\underbrace{K(\text{sel})}_{:=\,\lceil \log i_y^\pi(x)\rceil}\ +\ O(1).
\]

\begin{theorem}[Description plus selection]\label{thm:descsel}
For a universal executor $E$ and every $x\in D_V$,
\begin{equation}
\Mmin(x)\ \le\ \min_{y:\,C(y)=1}\ \Bigl(K(y)\ +\ \bigl\lceil\log i_y^\pi(x)\bigr\rceil\Bigr)\ \pm O(1).
\end{equation}
Moreover, for each fixed feasible $y$ one has the conditional lower bound
\begin{equation}\label{eq:condLB}
\min_{r:\,\Phi(x,r)=y}K(r)\ \ge\ K(y\mid x)\ \pm O(1),
\end{equation}
and the upper bound
\begin{equation}\label{eq:UB}
\min_{r:\,\Phi(x,r)=y}K(r)\ \le\ K(y)+\bigl\lceil\log i_y^\pi(x)\bigr\rceil\ \pm O(1).
\end{equation}
\end{theorem}

\begin{remark}[Unconditional lower bound under decoding]
If there exists a total computable $\Psi$ (independent of $x$) and a total computable $G$ such that for all feasible $(x,r)$, $G(\Psi(r))$ equals the part of $y$ used in the bound, then $K(r)\ge K(y)-O(1)$ along the fiber, so \eqref{eq:condLB} strengthens to $\min_{r:\Phi(x,r)=y}K(r)\ge K(y)\pm O(1)$.
\end{remark}

Fibers may contain accidentally simple realizers with $K(r)\ll K(y)+\log i_y^\pi(x)$, so equality requires structure; when such shortcuts are absent, the bound is tight in a precise sense. A first tightness regime arises under finite ambiguity: if for every feasible $y$ there exists a finite computable prototype set $R_y$ such that whenever $y$ is feasible at $x$ some $r\in R_y$ realizes it, then $\Mmin(x)=\min_{y:\,C(y)=1}K(y)\pm O(1)$ and NAQ reduces to the percentile of $K(y^\ast(x))$ across the pool. A second regime can be captured by an explicit genericity assumption.

\begin{assumption}[Fiber genericity]\label{ass:generic}
There exist constants $c,\alpha>0$ such that for all $x$ and all feasible $y$, the number of $r$ in the fiber $\{r:\Phi(x,r)=y\}$ with $K(r)\le K(y)+\lceil\log i_y^\pi(x)\rceil-c$ is at most $i_y^\pi(x)^{-\alpha}$.
\end{assumption}

\begin{remark}[Natural sufficient condition]
Assumption~\ref{ass:generic} holds if the universal semimeasure has no heavy concentration on compressed outliers within each fiber, e.g.
\[
\sum_{\substack{r:\ \Phi(x,r)=y\\ K(r)\le K(y)+\lceil\log i_y^\pi(x)\rceil-c}}
2^{-K(r)}\ \le\ i_y^\pi(x)^{-\alpha}\,2^{-K(y)-\lceil\log i_y^\pi(x)\rceil},
\]
which is implied by $K(r)\gtrsim K(y)+\lceil\log\mathrm{rank}_\pi(r)\rceil$ up to $O(1)$ on the fiber (no exponentially many compressed outliers).
\end{remark}

\begin{proposition}[Tightness under genericity]
If Assumption~\ref{ass:generic} holds, then
\begin{equation}
\Mmin(x)\ =\ \min_{y}\ \Bigl(K(y)+\bigl\lceil\log i_y^\pi(x)\bigr\rceil\Bigr)\ \pm O(1).
\end{equation}
\end{proposition}

A final piece of robustness concerns the choice of enumeration. If two computable bijections $\pi$ and $\pi'$ distort indices by at most a multiplicative factor $D$ on the pairs $(x,y)$ that attain the minimum, namely
\begin{equation}
\frac{1}{D}\ \le\ \frac{i_y^{\pi'}(x)}{i_y^{\pi}(x)}\ \le\ D,
\end{equation}
then
\begin{equation}
\min_y\Bigl[K(y)+\bigl\lceil\log i_y^{\pi'}(x)\bigr\rceil\Bigr]\ =\
\min_y\Bigl[K(y)+\bigl\lceil\log i_y^{\pi}(x)\bigr\rceil\Bigr]\ \pm \bigl\lceil\log D\bigr\rceil\ \pm O(1),
\end{equation}
so, after coarse binning, NAQ ranks are unchanged across enumerations.

\section{An NP-style regime with a linear worst-case bound}\label{sec:poly}
Fix $n\ge2$ and let $\mathcal X_n=\{x_s:\ s\in\{0,1\}^n\}$ where $x_s$ self-delimits $s$; responses are pairs $r=\langle s',y\rangle$, and $V_{\mathrm P}(x_s,r)=1$ if and only if $s'=s$ and $y=\text{``ok''}$, so the feasible sets $F_s$ are pairwise disjoint and the input-blindness of the executor forces the following linear lower bound.

\begin{proposition}[Linear worst-case advice under input blindness]\label{prop:linear}
For any total input-blind $E$ and any $n\ge2$,
\begin{equation}
\sup_{x\in\mathcal X_n} C^{\mathrm{adv}}_{\rxn,E}(x)\ \ge\ n-O(1).
\end{equation}
\end{proposition}
\begin{proof}
If every $x\in\mathcal X_n$ admitted advice shorter than $n-c$, then the number of distinct advice strings would be $<2^{n-c}$, hence two distinct identifiers would be served by the same advice, contradicting the disjointness of feasible sets.
\end{proof}

\begin{remark}[NAQ view]
Against the pool $\mathcal T=\mathcal X_n$, the top NAQ quantiles near~1 must be occupied since some instances necessarily incur $\Omega(n)$ advice; cf.\ instance and advice complexity~\cite{OKSW1994,FortnowKummer1996,BockenhauerKServer2017,BoyarAdaptiveAdvice2024}.
\end{remark}

\section{Resource-bounded variants and \texorpdfstring{$\NAQt$}{NAQ\_t}}\label{sec:rb}

\paragraph{Notation dictionary.} For an executor $E$ and advice $w$, let $\tau_E(w)$ denote the halting time of $E$ on $w$. For an output string $r$, let
\[
\tau^\star(r)\ :=\ \min\bigl\{\,\tau_U(p):\ U(p)=r,\ |p|=K(r)\,\bigr\}
\]
be the running time of a shortest $U$-program for $r$.

In the spirit of Levin’s time-penalized description length $Kt$~\cite{Levin1973} and universal search~\cite{UniversalSearch}, for a time-constructible budget $T$ define the truncated advice$+$time objective
\begin{equation}
C_{\rxn,E}^{(T)}(x)\ =\ \min \Bigl\{|w|+\log\bigl(1+\tau_E(w)\bigr):\ \exists r\ \bigl(E(w)=r\bigr)\ \wedge\ V_T(x,r)=1\Bigr\},
\end{equation}
where $V_T$ is verified within time $T(|x|)$.

\begin{lemma}[Exact truncation]\label{lem:trunc}
Assume a computable bound $B$ with the property that for all $x\in D_V$ there exists $w$ such that $E(w)=r$, $V_T(x,r)=1$, and $|w|+\log(1+\tau_E(w))\le B(|x|)$. Then $C^{(T)}_{\rxn,E}(x)$ is uniformly computable by simulating each $w$ with $|w|\le B(|x|)$ for
\begin{equation}
\Theta(x,w)\ :=\ 2^{\,B(|x|)-|w|}-1
\end{equation}
steps, and the cutoff is tight in the following bidirectional sense:
\begin{align*}
\text{\emph{(discard)}}\quad &\tau_E(w)>\Theta(x,w)\ \Longrightarrow\ |w|+\log(1+\tau_E(w))>B(|x|),\\
\text{\emph{(keep)}}\quad &\tau_E(w)\le \Theta(x,w)\ \Longrightarrow\ |w|+\log(1+\tau_E(w))\le B(|x|).
\end{align*}
\end{lemma}

\noindent\emph{Additive distortion under truncation.}
When advice length is a priori bounded by $B(|x|)$, the additive distortion incurred by truncation in the objective $|w|+\log(1+\tau_E(w))$ is $O(1)$ (in fact zero under the stated cutoff), i.e., the computed value coincides with the true minimum up to a constant independent of $x$.

\begin{definition}[Resource-bounded quantile $\NAQt$]
Let $M_T(x):=\min_{r:\,V_T(x,r)=1}\bigl(K(r)+\log(1+\tau^\star(r))\bigr)$. For a pool $\mathcal{T}$, define $\NAQt(x;\mathcal{T})$ as the percentile of $M_T(x)$ over $\mathcal{T}$; the coarse invariance and pool-stability statements for NAQ transfer \emph{mutatis mutandis}.
\end{definition}

\section{Operational converse: calibrating NAQ mass}\label{sec:converse}
Consider any (possibly randomized) mechanism that, on $x\sim P$, emits prefix-free advice $W$ for an input-blind executor producing $\hat r$ and $\hat y=\Phi(x,\hat r)$; fix a canonical $y^*(x)$ (for example, the $\prec$-least minimizer in Theorem~\ref{thm:descsel}) and require $\Pr[\hat y\neq y^*(x)]\le\varepsilon$. Writing $L=\mathbb{E}[|W|]$ and $Y^*=y^*(X)$, a Fano/Rate--Distortion converse (e.g., \cite{CoverThomas2006}) yields  
\noindent\emph{Probability space.}  
Let $(\Omega,\mathcal{F},\mathbb{P})$ be a probability space, $X\sim P$ the random input,  
$W$ the (random) prefix-free advice emitted by the encoder, $\hat r=\hat r(W)$ the executor's output,  
$\hat Y=\Phi(X,\hat r)$, and $Y^*=y^*(X)$. All probabilities and entropies are with respect to $\mathbb{P}$.  
\begin{theorem}[Fano/Rate--Distortion converse]\label{thm:fano}  
Assume $\mathcal{Y}$ is finite (equivalently $H(Y^*)<\infty$). Then  
\begin{equation}  
L \ge H(Y^*) - h(\varepsilon) - \varepsilon\log(|\mathcal{Y}|-1) - O(1),
\end{equation}  
where the $O(1)$ arises from Kraft--McMillan: for prefix-free $W$, $H(W)\le \mathbb{E}[|W|]+1$.  
\end{theorem}  
\begin{remark}[NAQ interpretation]  
Under finite ambiguity, $\Mmin(x)=K\bigl(Y^*(x)\bigr)\pm O(1)$, hence pushing probability mass to high NAQ quantiles forces large $H(Y^*)$ and therefore a large average advice length $L$ by the converse.  
\end{remark}

\section{Empirical estimation of NAQ (DKW)}\label{sec:dkw}
Let $F$ be the cdf of $\Mmin(X)$ under a computable source $P$, and let $\widehat F_n$ be the empirical cdf over $n$ i.i.d.\ samples $X_1,\dots,X_n\sim P$ (or over a computable pool). The Dvoretzky--Kiefer--Wolfowitz inequality~\cite{DKW1956} gives a sharp uniform guarantee:
\begin{theorem}[Dvoretzky--Kiefer--Wolfowitz]\label{thm:dkw}
\begin{equation}
\Pr\Bigl(\ \sup_{z}\,|\widehat F_n(z)-F(z)|>\epsilon\ \Bigr)\ \le\ 2e^{-2n\epsilon^2}\qquad\text{for all }\epsilon>0,
\end{equation}
so the empirical NAQ converges uniformly to the true quantile, and coarse binning further stabilizes ranks across universal machines (Theorem~\ref{thm:coarse}).
\end{theorem}

In practice one approximates $\Mmin$ by compressor-based proxies $\widehat K$; the DKW bound still controls sampling error, while model error is handled by compressor choice and cross-checks.

\section{Related work}
Connections span certificate and decision-tree complexity~\cite{BuhrmanDeWolf2002}; witness models (NP, MA)~\cite{AroraBarak2009,Sipser2012}; algorithmic statistics and MDL~\cite{VereshchaginVitanyi2004,VereshchaginVitanyi2008,Grunwald2007}; advice and instance complexity~\cite{OKSW1994,FortnowKummer1996,BockenhauerKServer2017,BoyarAdaptiveAdvice2024}; information distance~\cite{BennettEtAl1998}; conditional semimeasures and the coding theorem~\cite{LiVitanyi2019}; and time-penalized description length in the spirit of Levin’s $Kt$~\cite{Levin1973,UniversalSearch}. For computable analysis of $\arg\min$ see Weihrauch and subsequent developments~\cite{Weihrauch2000,BrattkaPauly2011}. Compression-based similarity via normalized compression distance (NCD)~\cite{CilibrasiVitanyi2005} is related to description length yet orthogonal to our input-blind advice lens.

\section{Biological application: a formal theory for adaptive immunity}\label{sec:bio}
We state the formal results with concise effective hypotheses; expanded modeling details are in Appendix~\ref{app:bio-expanded}.

\paragraph{Effective search universe (concise).}
Blueprints live in a computably enumerable domain over finite alphabets with computable length bounds; surrogate binding/presentation maps $(\widehat A_{\mathrm{B}},\widehat A_{\mathrm{MHC}},\widehat A_{\mathrm{T}})$ and thresholds are total computable; proteasomal/endosomal digestions and MHC-binding filters are total computable transductions; the within-host dynamics $\mathcal D$ is a finite-horizon computable system with rational coefficients.

\begin{proposition}[Computable presentation and semi-decidable recognition]\label{prop:presentation-computable}
Under the effective search universe above, the presentation map
\begin{equation}
\mathcal P_H:\ (\mathrm{Seq},\mathrm{PTM},\mathrm{Comp})\ \longmapsto\ \bigl\{\,p:\ \widehat A_{\mathrm{MHC}}(p,H)\ge \theta_{\mathrm{MHC}}\,\bigr\}
\end{equation}
is total and computable. Consequently $\mathrm{Recog}(x,r)$ is c.e.\ and $\mathrm{Exec}(x,r)$ is decidable. If safety is enforced by typed blueprints, then $V$ is c.e.; if safety is checked by explicit universal quantification over a computable self set $\mathcal S$, then $V\in\Delta^0_2$.
\end{proposition}

\begin{lemma}[Tolerance barrier $\Rightarrow$ finite ambiguity]\label{lem:tolerance-finite}
Let $d$ be a computable metric on epitope/pMHC descriptors and suppose: (i) for each feasible $y$ there is a computable length bound $L(y)$ for descriptor strings; (ii) a tolerance margin $\delta>0$ with $\min_{s\in\mathcal S} d(e,s)\ge \delta$ for all pathogen epitopes $e$ realizable in $x$; (iii) typed blueprint guards forbid recognition above threshold whenever $d(e,s)<\delta$. Then for every feasible $y$ there exists a finite computable prototype set $R_y$ such that whenever $y$ is feasible at $x$ some $r\in R_y$ realizes it, and consequently
\begin{equation}
\Mmin(x)\ =\ \min_{y:\,C(y)=1}K(y)\ \pm O(1).
\end{equation}
\end{lemma}

\begin{theorem}[GC selection $\Leftrightarrow$ selection hardness]\label{thm:gc-selection}
Model a germinal-center reaction as generating $N_t$ candidates by time $t$, tested i.i.d.\ with success probability $p_y(x)$ for realizing some $r$ in the fiber $\mathcal R_y(x)=\{r:\Phi(x,r)=y\}$. Then
\begin{equation}
\Pr\!\bigl(S_t\,\big|\,N_t=n\bigr)\ =\ 1-(1-p_y(x))^{n},\qquad
\Pr(S_t)\ \le\ 1-(1-p_y(x))^{\,\mathbb{E}[N_t]},
\end{equation}
so $\mathbb{E}[N_t]\ge \frac{1}{p_y(x)}\log\!\frac{1}{\varepsilon}$ is necessary for $\Pr(S_t)\ge 1-\varepsilon$. If $\pi$ is length-lex and candidates are drawn with universal semimeasure mass $q(r)\asymp 2^{-K(r)}$, then
\begin{equation}
\Bigl|\ \bigl\lceil\log i_y^\pi(x)\bigr\rceil\ -\ \bigl\lceil\log\!\bigl(1/p_y(x)\bigr)\bigr\rceil\ \Bigr|\ \le\ O(1),
\end{equation}
with constants depending only on the choice of universal machine.
\end{theorem}

\begin{theorem}[Variant-panel lower bound via overlap]\label{thm:variant-lb}
Let $S\subseteq\mathcal X$ be $\rho$-separated in the sense that feasibility implies $d(\mathrm{ctx}(r),x)\le \rho$ while $d(x,x')>2\rho$ for $x\neq x'$. Then $\Delta(S):=\max_r |\{x\in S:\ r\in F(x)\}|=1$, and
\begin{equation}
\max_{x\in S}\ \Mmin(x)\ \ge\ \bigl\lceil \log|S| \bigr\rceil\ -\ O(1).
\end{equation}
\end{theorem}

\begin{theorem}[Targeted incompressibility for deadlines]\label{thm:time-to-control}
Let $V_T$ require first-control $t^{\ast}\le T$ and suppose there is a total computable $\Gamma(x,r)=G(\Phi(x,r))$ extracting an $n(x)$-bit kinetic guard that must be recoverable under $V_T$. Then
\begin{equation}
M_T(x)\ \ge\ \min_{y\in\mathcal{Y}_T} K\bigl(G(y)\bigr)\ \pm O(1)\ \ge\ n(x)\ \pm O(1),
\qquad \mathcal{Y}_T=\{y:\ C(y)=1,\ t^{\ast}(y)\le T\}.
\end{equation}
\end{theorem}

\begin{theorem}[Memory rate--distortion bound for recall]\label{thm:memory-rd}
Let $Y^{\ast}(x)$ be the canonical optimal certificate (e.g.\ the $\prec$-least minimizer). Model immune memory as a prefix-free message $W$ available at recall; any decoder randomization is absorbed into $W$. If a recall policy achieves $\Pr\!\bigl[\hat y\neq Y^{\ast}(X)\bigr]\le \varepsilon$ with $\mathbb{E}[|W|]=B$, then
\begin{equation}
B\ \ge\ H(Y^{\ast})\ -\ h(\varepsilon)\ -\ \varepsilon\log(|\mathcal{Y}|-1)\ -\ O(1),
\end{equation}
and the bound holds for block or variable-length prefix-free codes (optimality not assumed).
\end{theorem}

\noindent These results show how NAQ, and its timed variant $\NAQ_t$, arise from concrete mechanisms: clonal exploration materializes the selection term; processing/presentation keep feasibility effective; tolerance margins yield finite ambiguity so descriptive cost dominates; variant panels force worst-case hardness; kinetic deadlines push difficulty upward; and memory imposes an information budget.

\section{Consequences}
In computably enumerable regimes, the advice functional $C^{\mathrm{adv}}$ aligns with Kolmogorov complexity and therefore inherits its noncomputability. When validity is decidable and the loss is computable, the optimization problem collapses in the sense that $r^{\ast}$ is uniformly computable and $K(r^{\ast}\!\mid x)$ remains bounded. If validity is decidable but the loss is only upper-semicomputable, $r^{\ast}$ may fail to be computable; if validity is merely c.e.\ while the loss is computable, noncomputability already enters through feasibility. Under resource bounds, however, one can recover exact truncation provided an effective \emph{a priori} budget is available, and NAQ encapsulates all of these separations into a universal hardness scale whose calibration is supplied directly by the operational converse. See also the computable analysis perspective on $\arg\min$ and selection~\cite{Weihrauch2000,BrattkaPauly2011}.

\section{Open problems}
Limit laws for NAQ under canonical task ensembles; NAQ-driven hypothesis tests and meta-selection procedures; a refined placement of argmin and fiber selection in the Weihrauch lattice; average-case bounds on $\mathbb{E}[\lceil\log i_{Y^{\ast}}(X)\rceil]$ and principled tradeoffs between description and selection; resource-bounded analogues based on Levin’s $Kt$; and, for biology in particular, principled computable feature maps that extend Appendix~\ref{app:bio-expanded}.

\appendix

\section{Proof of Theorem~\ref{thm:HA-decidable} (decidable $V$ and computable $L$)}\label{app:HA-dec}
\begin{proof}
Fix a recursive, strictly increasing enumeration $\Lambda=\{\lambda_0<\lambda_1<\cdots\}\subset\mathbb{Q}_{\ge 0}$ containing the range of $L$, and a computable $\prec$-enumeration $(e_k)_{k\ge 0}$ of $\Gamma^\ast$. For $x\in D_V$ define the (nonempty) index set
\[
J(x):=\{\,j\in\mathbb{N}:\ \exists r\in\Gamma^\ast\ \text{s.t. }V(x,r)=1\ \wedge\ L(x,r)\le \lambda_j\,\}.
\]
Since $J(x)\subseteq\mathbb{N}$ is nonempty, it has a \emph{least} element $j_0=\min J(x)$. Define
\[
K(x):=\{\,k\in\mathbb{N}:\ V(x,e_k)=1\ \wedge\ L(x,e_k)\le \lambda_{j_0}\,\}.
\]
By definition of $j_0$, $K(x)$ is nonempty. Because $\prec$ is a computable total order with enumeration $(e_k)$, the $\prec$-least $r$ with loss $\le\lambda_{j_0}$ is exactly $e_{k^\ast}$ where $k^\ast=\min K(x)$.

We now give a \emph{terminating} procedure that, on input $x$, returns $e_{k^\ast}$.

\smallskip
\noindent\emph{Algorithm (dovetail on $(j,k)$):} For stages $s=0,1,2,\dots$, scan all pairs $(j,k)$ with $j\le s$ and $k\le s$ in lexicographic order on $(j,k)$; for each pair, compute (i) $V(x,e_k)$ (decidable by assumption) and (ii) $L(x,e_k)$ (computable by assumption); if $V(x,e_k)=1$ and $L(x,e_k)\le \lambda_j$, halt and \emph{output} $e_k$.

\smallskip
\noindent\emph{Termination.} Since $j_0\in J(x)$ and $k^\ast\in K(x)$, the pair $(j_0,k^\ast)$ will be examined at stage $s\ge\max\{j_0,k^\ast\}$. When it is examined, both predicates are true and the algorithm halts. Hence the algorithm terminates on every $x\in D_V$.

\smallskip
\noindent\emph{Correctness.} Suppose the algorithm halts on $(j,k)$. By construction, $L(x,e_k)\le\lambda_j$. If $j>j_0$, then, because pairs are scanned in increasing $(j,k)$, the pair $(j_0,k^\ast)$ would have been scanned earlier and would have caused termination; contradiction. Thus $j=j_0$. Among those with $j=j_0$, pairs are scanned in increasing $k$, so $k$ must be the \emph{least} index with $V(x,e_k)=1$ and $L(x,e_k)\le \lambda_{j_0}$, i.e.\ $k=k^\ast$. Therefore the output is exactly
\[
r^\ast(x)=e_{k^\ast}\in\arg\min\nolimits_{\prec}\{\,L(x,r):V(x,r)=1\,\}.
\]

\smallskip
\noindent\emph{Uniform computability and bounded conditional complexity.} The procedure above is a single Turing machine that, on input $x$, halts with $r^\ast(x)$. Hence $x\mapsto r^\ast(x)$ is total computable on $D_V$. Let $A$ denote a fixed description (program) of this machine. For all $x$, the conditional Kolmogorov complexity satisfies
\[
K\bigl(r^\ast(x)\mid x\bigr)\ \le\ |A|\ +\ O(1)\ =: C,
\]
because a shortest program can consist of a fixed code for $A$ plus a constant-size wrapper that forwards its conditional input $x$ to $A$ and prints the result. The constant $C$ is global (independent of $x$). This proves the theorem.
\end{proof}

\section{Proof of Theorem~\ref{thm:HA-decidable-usc} (decidable $V$ and upper-semicomputable $L$)}\label{app:HA-usc}
\begin{proof}
We work under the stated promise that inputs are of the form $x=0^n1u$, from which $n$ is computably parsed. Fix a computable, prefix-free encoding $\mathrm{enc}(n,b)$ for $(n,b)\in\mathbb{N}\times\{0,1\}$, and define
\[
r_{n,0}:=\mathrm{enc}(n,0),\qquad r_{n,1}:=\mathrm{enc}(n,1).
\]
Let the acceptance predicate be decidable and \emph{trivial}: $V(x,r)\equiv 1$ for all $(x,r)$.

Define the loss $L:\Sigma^\ast\times\Gamma^\ast\to\mathbb{Q}_{\ge 0}$ by
\[
L(x,r)=
\begin{cases}
2n+1-\mathbf{1}\{\text{$p_n$ halts}\}, & \text{if } r=r_{n,1},\\
2n, & \text{if } r=r_{n,0},\\
2n+2+|r|, & \text{otherwise,}
\end{cases}
\]
where $(p_n)_{n\ge 1}$ is any fixed prefix-free enumeration of programs for a universal machine. Then:
\begin{itemize}
\item The \emph{codomain} of $L$ is a recursive discrete subset of $\mathbb{Q}$ (namely $\{2n,2n+1\}\cup\{2n+2+|r|:r\in\Gamma^\ast\}$).
\item $L$ is \emph{upper-semicomputable}: for $r=r_{n,1}$ output the rational sequence $q_s(x,r)\equiv 2n+1$ for all $s$ until a halting witness for $p_n$ is enumerated, at which point switch forever to $q_s(x,r)\equiv 2n$; for $r=r_{n,0}$ and all other $r$ output the constant sequences $2n$ and $2n+2+|r|$, respectively. In all cases $q_s(x,r)\searrow L(x,r)$.
\item For every $x=0^n1u$ and every $r\notin\{r_{n,0},r_{n,1}\}$ one has $L(x,r)=2n+2+|r|\ge 2n+2>2n+1\ge L(x,r_{n,1})$, so \emph{no} such $r$ can be a minimizer.
\end{itemize}
Consequently the set of minimizers at $x$ is contained in $\{r_{n,0},r_{n,1}\}$. We distinguish two cases.

\smallskip
\noindent\emph{Case 1: $p_n$ does not halt.} Then $L(x,r_{n,1})=2n+1>L(x,r_{n,0})=2n$, so the unique minimizer is $r^\ast(x)=r_{n,0}$.

\smallskip
\noindent\emph{Case 2: $p_n$ halts.} Then $L(x,r_{n,1})=L(x,r_{n,0})=2n$. By hypothesis, the global tie-breaking order $\prec$ is fixed so that, at equal loss, $r_{n,1}\prec r_{n,0}$. Therefore $r^\ast(x)=r_{n,1}$.

\smallskip
\noindent\emph{Noncomputability of the optimizer.} Consider the map $f:\mathbb{N}\to\{0,1\}$ defined by
\[
f(n)\ :=\ \begin{cases}
1, & \text{if } r^\ast(0^n1u)=r_{n,1},\\
0, & \text{if } r^\ast(0^n1u)=r_{n,0}.
\end{cases}
\]
By the case analysis above, $f(n)=\mathbf{1}\{\text{$p_n$ halts}\}$. Thus if $x\mapsto r^\ast(x)$ were computable, then the halting problem would be decidable---a contradiction. Hence $x\mapsto r^\ast(x)$ is \emph{not} computable.

\smallskip
\noindent\emph{Remark on conditional complexity.} For each fixed $x=0^n1u$, a constant-size program can, given $x$ as conditional input, parse $n$ and print either $r_{n,0}$ or $r_{n,1}$ with the choice hard-coded; therefore $K(r^\ast(x)\mid x)=O(1)$ pointwise (this does not contradict noncomputability of the \emph{function} $x\mapsto r^\ast(x)$).
\end{proof}

\section{Proof of Theorem~\ref{thm:HA-ce} (c.e.\ $V$ and computable $L$)}\label{app:HA-ce}
\begin{proof}
Again work under the promise $x=0^n1u$ with $n$ computably parsed, and retain the prefix-free encodings $r_{n,0},r_{n,1}$ as above. Define a \emph{computably enumerable} acceptance predicate $V$ by
\[
V(x,r)=
\begin{cases}
1, & \text{if } r=r_{n,1}\ \text{and}\ p_n\ \text{halts},\\
1, & \text{if } r=r_{n,0},\\
0, & \text{otherwise.}
\end{cases}
\]
That $V$ is c.e.\ follows because the set $\{n:\ p_n\ \text{halts}\}$ is c.e., hence so is the relation $\{(x,r): V(x,r)=1\}$.

Define a \emph{total computable} discrete loss by
\[
L(x,r)=
\begin{cases}
0, & \text{if } r=r_{n,1},\\
1, & \text{if } r=r_{n,0},\\
2+|r|, & \text{otherwise.}
\end{cases}
\]
Then for any $x=0^n1u$:

\smallskip
\noindent\emph{Case 1: $p_n$ halts.} Both $r_{n,1}$ and $r_{n,0}$ are feasible, with losses $0$ and $1$ respectively; all other $r$ are either infeasible or have loss $\ge 2$. The unique minimizer is $r^\ast(x)=r_{n,1}$.

\smallskip
\noindent\emph{Case 2: $p_n$ does not halt.} The only feasible element among $\{r_{n,0},r_{n,1}\}$ is $r_{n,0}$ (since $V(x,r_{n,1})=0$); all other $r$ either have $V(x,r)=0$ or loss $\ge 2$. Thus the unique minimizer is $r^\ast(x)=r_{n,0}$.

\smallskip
\noindent\emph{Noncomputability of the optimizer.} Define $f(n)$ as in the previous proof:
\[
f(n)=\begin{cases}
1, & \text{if } r^\ast(0^n1u)=r_{n,1},\\
0, & \text{otherwise.}
\end{cases}
\]
Then $f(n)=\mathbf{1}\{\text{$p_n$ halts}\}$. Hence, if $x\mapsto r^\ast(x)$ were computable, the halting problem would be decidable, which is impossible. Therefore $x\mapsto r^\ast(x)$ is not computable.

\smallskip
\noindent\emph{Remark on conditional complexity.} As before, for each fixed $x$ a constant-length conditional program can parse $n$ from $x$ and print $r_{n,0}$ or $r_{n,1}$ as hard-coded; thus $K(r^\ast(x)\mid x)=O(1)$ pointwise, even though the \emph{mapping} $x\mapsto r^\ast(x)$ is not computable.
\end{proof}

\section{Expanded biological modeling details (moved from Sec.~\ref{sec:bio})}\label{app:bio-expanded}
Within this encoding, an antigenic context is a finite record $x=(\mathrm{Seq};\mathrm{PTM};\mathrm{Comp};\mathrm{Kinetics};\mathrm{Host})$, where $\mathrm{Seq}$ is a finite multiset of protein or peptide strings representing pathogen proteome fragments, $\mathrm{PTM}$ records post-translational modifications and conformational motifs such as disulphides and glycan masks that govern BCR epitope visibility, $\mathrm{Comp}$ gives the compartmental footprint (extracellular, endosomal, cytosolic) that determines MHC class routing (II versus I), $\mathrm{Kinetics}$ specifies a simple growth/clearance model with parameters $(\beta,\delta)$ and an initial inoculum $n_0$, and $\mathrm{Host}$ lists MHC alleles $H=(H^{\mathrm{I}},H^{\mathrm{II}})$ together with a tolerance profile $\mathcal{S}$ that abstracts the self peptidome; all fields admit computable encodings once finite alphabets and rational parameters are fixed. A response is a self-delimiting code $r=(\mathsf{Rec},\mathsf{Eff},\mathsf{Mem},\mathsf{Sched})$ decoded by a total input-blind interpreter: $\mathsf{Rec}$ is a finite set of clonotype blueprints for BCR and/or TCR detailing V(D)J usage, junctional constraints (regex-style), and an affinity target band; $\mathsf{Eff}$ is a finite vector over a small alphabet of effector modes such as $\mathrm{CD8\text{-}CTL}$, $\mathrm{Th1}$, $\mathrm{Th2}$, $\mathrm{Th17}$, $\mathrm{Tfh}$, $\mathrm{Treg}$, and $\mathrm{PlasmaB}$ with cytokine knobs; $\mathsf{Mem}$ prescribes memory allocations for B, CD4, and CD8 lineages with half-life parameters; and $\mathsf{Sched}$ is a coarse schedule $\{(t_i,\mathrm{module}_i)\}$ with discrete days $t_i\in\{0,\dots,T_{\max}\}$.

Validity is the conjunction of recognition, safety, and execution: we set $V(x,r)=1$ if and only if $\mathrm{Recog}(x,r)\wedge \mathrm{Safe}(x,r)\wedge \mathrm{Exec}(x,r)$, and we display this compactly as
\begin{equation}
V(x,r)\ =\ \underbrace{\mathrm{Recog}(x,r)}_{\text{binding \& presentation}}\ \wedge\ \underbrace{\mathrm{Safe}(x,r)}_{\text{tolerance}}\ \wedge\ \underbrace{\mathrm{Exec}(x,r)}_{\text{operational success}}.
\end{equation}
Recognition holds if either there exists an epitope $e$ formed by $(\mathrm{Seq},\mathrm{PTM})$ for which some blueprint in $\mathsf{Rec}$ admits a BCR conformer whose coarse-grained binding functional $A_{\mathrm{B}}(e)$ exceeds a threshold $\theta_{\mathrm{B}}$, or there exists a peptide $p$ in the proteasomal or endosomal digest of $\mathrm{Seq}$ that binds an allele in $H$ with $A_{\mathrm{MHC}}(p,H)\ge \theta_{\mathrm{MHC}}$ and for which some blueprint in $\mathsf{Rec}$ admits a TCR conformer with $A_{\mathrm{T}}(p{:}H)\ge \theta_{\mathrm{T}}$; the binding maps $A_{\mathrm{B}},A_{\mathrm{MHC}},A_{\mathrm{T}}$ are taken as computable surrogates, so the existential nature of the check makes $\mathrm{Recog}$ c.e. Safety requires that no self peptide $s\in\mathcal{S}$ with sufficiently strong presentation is recognized above a tolerance band, which we write as
\begin{equation}
\neg\exists s\in\mathcal{S}:\ A_{\mathrm{MHC}}(s,H)\ge \theta_{\mathrm{MHC}}^{\mathrm{self}}\ \wedge\ A_{\mathrm{T}}(s{:}H)\ge \theta_{\mathrm{T}}^{\mathrm{self}},
\end{equation}
with an analogous guard for BCR via a library of self-like conformational motifs; being universal over an effectively enumerable $\mathcal{S}$, this is typically co-c.e.\ (violations are c.e.). Execution requires that when the computable within-host dynamical system $\mathcal{D}$, a rational-coefficient ODE/agent hybrid on an integer time grid, is driven by $\mathsf{Eff}$ and $\mathsf{Sched}$ instantiated by the recognized module(s), the pathogen burden remains below a harm threshold throughout the horizon:
\begin{equation}
\forall t\le T_{\max}:\ \mathrm{Burden}_{\mathcal{D}}(t)\ \le\ B_{\mathrm{safe}}.
\end{equation}
Because $\mathcal{D}$ is computable and $T_{\max}$ is finite, $\mathrm{Exec}$ is decidable once recognition anchors the inputs; overall, $V$ is c.e.\ if safety is enforced by typed blueprints that exclude self motifs by construction, while insisting on explicit safety checking over $\mathcal{S}$ makes $V$ a conjunction of a c.e.\ and a co-c.e.\ property (hence $\Delta^0_2$ in general).

\subsection*{Case families and design heuristics (details)}
Three indicative regimes suffice to illustrate the encoding: intracellular infections with fast kinetics (cytosolic replication, tight deadlines, $\mathrm{CD8\text{-}CTL}$ modes, time penalty $\log(1+\tau)$ raising $M_T$; finite-ambiguity collapse when conserved peptides exist), extracellular settings with glycan-shielded antigens (sparse BCR epitopes; small finite prototype sets $R_y$ make $\Mmin\approx \min_y K(y)$), and autoimmunity-prone tissues (binding-risk constraints prune blueprints, increasing $K(y)$ and potentially $\lceil\log i^\pi_y(x)\rceil$, shifting NAQ upward).

\subsection*{Design implications (details)}
When finite prototype sets $R_y$ exist, shape the context $x$ to surface low-$K(y)$ certificates with a safety margin, as selection cost collapses to $O(1)$. Otherwise act on selection hardness by constraining fibers (presentation structure, pruning), thereby reducing $\lceil \log i_y^\pi(x)\rceil$ without compromising safety. Under explicit resource constraints (small $T$), these choices interact transparently with $M_T$ and $\NAQ_t$.

\end{document}